\documentclass[11pt]{amsart}

\usepackage{amsmath, amssymb, amsthm, amscd}
\usepackage{url}
\usepackage{graphicx}
\usepackage[hidelinks,pagebackref,pdftex]{hyperref}
\usepackage{color}
\usepackage{import}
\usepackage[margin=3cm]{geometry}

\usepackage{array}
\usepackage{epsfig}
\usepackage{latexsym}
\usepackage{url}
\usepackage{enumerate}
\usepackage{stmaryrd}
\usepackage{wrapfig}
\usepackage{lineno}

\usepackage{epsfig}
\usepackage{url}
\usepackage{subfig}

\usepackage{hyperref}

\newcommand{\Z}{\mathbb{Z}}
\newcommand{\Q}{\mathbb{Q}}
\newcommand{\R}{\mathbb{R}}
\newcommand{\K}{\mathcal{K}}
\newcommand{\Field}{\mathbb{F}}
\newcommand{\rank}{\mbox{rank}}
\newcommand{\ann}{\mathsf{a}}

\newcommand{\UF}{\mathcal{U}\hspace{-1pt}\mathcal{F}}

\newcommand{\SC}{\K_{\text{DS}}}

\newcommand{\C}{\mathcal{C}}

\newcommand{\AV}{{\mathcal{AV}}}

\usepackage{color}

\numberwithin{equation}{section}
\theoremstyle{plain}

\newtheorem{proposition}[equation]{Proposition}
\theoremstyle{definition}
\newtheorem{definition}[equation]{Definition}

\begin{document}
\title[The Compressed Annotation Matrix for Computing Persistent Cohomology]{The Compressed Annotation Matrix: an Efficient Data Structure
  for Computing Persistent Cohomology}
%
\author{Jean-Daniel Boissonnat}
\author{Tamal K. Dey}  
\author{Cl\'ement Maria}
%
\email{ \{jean-daniel.boissonnat, clement.maria\}@inria.fr}
\email{tamaldey@cse.ohio-state.edu}
%
%
\maketitle
%
%
\begin{abstract} 
Persistent homology with coefficients in a field $\Field$
coincides with the same for cohomology because of duality.
We propose
an implementation of a recently introduced algorithm for persistent
cohomology that attaches annotation vectors with the simplices.
We separate the representation of the simplicial complex from the
representation of the cohomology groups, and introduce a new data structure for maintaining
the annotation matrix, which is more compact and reduces substancially the amount
of matrix operations. In addition, we propose a heuristic to simplify further the representation of the
cohomology groups and improve both time and space complexities.
The paper provides 
a theoretical analysis,  as well as a detailed experimental
study of our implementation and comparison with state-of-the-art
software for persistent homology and cohomology.

\bigskip

This article appeared in Algorithmica 2015~\cite{DBLP:journals/algorithmica/BoissonnatDM15}. An extended abstract appeared in the proceedings of the European Symposium on Algorithms 2013~\cite{DBLP:conf/esa/BoissonnatDM13}.
\end{abstract}
%


%
\section{Introduction}
Persistent homology~\cite{DBLP:books/daglib/0025666} is an algebraic method for measuring the topological features of a space induced by the sublevel sets
of a function. Its generality and stability with regard to noise have made it a widely used tool for the study of data. A common approach is the study of the topological invariants of a nested family of simplicial complexes built on top of the data, seen as a set of points in a geometric space. This approach has been successfully used in various areas of science and engineering, as for example in sensor networks, image analysis, 
and data analysis 
where one typically needs to deal with big data sets in high dimensions.
Consequently, the demand for designing efficient algorithms and 
software to compute persistent homology of filtered simplicial 
complexes has grown.

The first persistence algorithm~\cite{DBLP:journals/dcg/EdelsbrunnerLZ02,DBLP:journals/dcg/ZomorodianC05} can be implemented
by reducing a matrix defined by face incidence relations, 
through column operations. The running time is $O(m^3)$ where $m$ is the number of simplices of the simplicial complex and, despite good performance in practice, Morozov proved that this bound is tight~\cite{Morozov05persistencealgorithm}. 
Recent optimizations taking advantage of the special structure of the  matrix to be reduced
have led to significant progress in the theoretical analysis~\cite{DBLP:journals/comgeo/ChenK13,DBLP:conf/compgeom/MilosavljevicMS11} as well as in practice~\cite{Bauer:arXiv1303.0477,Chen11persistenthomology}.

A different approach~\cite{DBLP:journals/corr/abs-1107-5665,DBLP:journals/dcg/SilvaMV11} interprets the persistent homology groups in terms of their dual, the persistent cohomology groups.  The cohomology algorithm has been reported to work better in practice than the standard homology algorithm~\cite{DBLP:journals/corr/abs-1107-5665} but this advantage seems to fade away when 
optimizations are employed to the homology algorithms~\cite{Bauer:arXiv1303.0477}. An elegant description of the cohomology algorithm, using the notion of annotations~\cite{DBLP:conf/swat/BusaryevCCDW12}, has been introduced in~\cite{DBLP:conf/compgeom/DeyFW14}
and used to design more general algorithms for maintaining cohomology groups under simplicial maps.

In this work, we propose an implementation of the annotation-based algorithm for computing persistent cohomology. 
A key feature of our implementation is a distinct separation between the representation of the simplicial complex and the representation of the cohomology groups. Currently the simplicial complex can be represented either by its Hasse diagram or by using the more compact simplex tree~\cite{DBLP:journals/algorithmica/BoissonnatM14}. The cohomology groups are stored in a separate data structure that represents a compressed version of the annotation matrix. As a consequence, the time and space complexities of our algorithm depend mostly on properties of the cohomology groups we maintain along the computation and only linearly on the size of the simplicial complex.

Moreover, maintaining the simplicial complex and the cohomology groups separately allows us
to reorder the simplices while keeping the same persistent cohomology. This significantly reduces the size of the cohomology groups to be maintained, and improves considerably both the time and memory performance as shown by our
detailed experimental analysis on a variety of examples.  Our method compares favorably with state-of-the-art software for computing persistent homology and cohomology.
%

\medskip

\noindent
{\bf Background:}
\label{sec:background}
A {\em simplicial complex} is a pair $\mathcal{K}=(V,S)$ where $V$ is
a finite set whose elements are called the {\em vertices} of $\mathcal{K}$ and
$S$ is a set of non-empty subsets of $V$ that is
required to satisfy the following two conditions~:
1.  $p\in V \Rightarrow \{  p\} \in S$ and 2. $\sigma\in S, \tau\subseteq \sigma \Rightarrow
\tau\in S$.
Each element $\sigma\in S$ is called a {\em simplex} or a \emph{face} of $\mathcal{K}$
and, if $\sigma\in S$ has precisely $s+1$ elements ($s\geq -1$),
$\sigma$ is called an $s$-simplex and its dimension is
$s$. The dimension of the simplicial complex $\mathcal{K}$ is
the largest $k$ such that $S$ contains a $k$-simplex. We define $\K^p$
to be the set of $p$-dimensional simplices of $\K$, and note its size
$|\K^p|$. Given two
simplices $\tau$ and $\sigma$ in $\K$, $\tau$ is a subface
(resp. coface) of $\sigma$ if $\tau \subseteq \sigma$ (resp. $\tau
\supseteq \sigma$). The \emph{boundary} of a simplex $\sigma$, denoted
$\partial \sigma$, is the set of its subfaces with codimension $1$.

A \emph{filtration}~\cite{DBLP:books/daglib/0025666} of a simplicial
complex is an order relation on its simplices which respects
inclusion. Consider a simplicial complex $\K=(V,S)$ and a function $\rho : S
\rightarrow \R$. We require $\rho$ to be monotonic in the sense that, for
any two simplices $\tau \subseteq \sigma$ in $\K$, $\rho$ satisfies
$\rho(\tau) \leq \rho(\sigma)$. We will call $\rho(\sigma)$ the
\emph{filtration value} of the simplex $\sigma$. Monotonicity implies that the sublevel
sets $\K (r) =
\rho^{-1}(-\infty,r]$ are subcomplexes of $\K$, for every $r \in \R$.
Let $m$ be the number of simplices of $\K$, 
and let $(\rho_i)_{i=1\cdots n}$ be the $n$ different values $\rho$ takes on
the simplices of $\K$. Plainly  $n \leq m$, and we have  the following  sequence of $n+1$
subcomplexes:

\[
\emptyset = \K_0 \subseteq \cdots \subseteq \K_n = \K, \hspace{0.2cm} 
-\infty = \rho_0 < \cdots < \rho_n, \hspace{0.2cm}  \K_i = \rho^{-1}(-\infty,\rho_i]
\]

\noindent
Applying a (co)homology functor to this sequence of simplicial
complexes turns 

\noindent
(combinatorial) complexes into (algebraic) abelian
groups and inclusion into group homomorphisms. Roughly speaking, a
simplicial complex defines a domain as an arrangement of local bricks
and (co)homology catches the global features of this domain, like the
connected components, the tunnels, the cavities, etc. The
homomorphisms catch the evolution of these global features when
inserting the simplices in the order of the filtration.
Let $H_p(\K)$ and $H^p(\K)$ denote respectively the homology and cohomology
groups of $\K$ of dimension
$p$ with coefficients in a field $\Field$. 
The filtration induces a sequence of homomorphisms
in the homology and cohomology groups in opposite directions: 

\begin{eqnarray}
0 = H_p(\K_0) \rightarrow H_p(\K_1) \rightarrow \cdots \rightarrow
H_p(\K_{n-1}) \rightarrow
H_p(\K_n) = H_p(\K) \label{hom}\\
0 = H^p(\K_0) \leftarrow H^p(\K_1) \leftarrow \cdots \leftarrow 
H^p(\K_{n-1}) \leftarrow 
H^p(\K_n) = H^p(\K) \label{cohom}
\end{eqnarray}

\noindent
We refer to~\cite{DBLP:books/daglib/0025666} for an introduction to
the theory of homology and
persistent homology. Computing the persistent
homology of such a sequence consists in pairing each simplex that creates a
homology feature with the one that destroys it. The usual output is a
\emph{persistence diagram}, which is a plot of the points
$(\rho(\tau),\rho(\sigma))$ for each persistent pair $(\tau,\sigma)$. 
It is known that because of duality the homology and cohomology sequences above
provide the same persistence diagram~\cite{DBLP:journals/dcg/SilvaMV11}.\\ 
\indent The original persistence 
algorithm~\cite{DBLP:journals/dcg/EdelsbrunnerLZ02} considers the
homology sequence in Equation~\ref{hom} that aligns with the filtration
direction. It detects when a new homology class is born and
when an existing class dies as we proceed forward through the filtration. 
Recently, a few algorithms have considered the cohomology sequence
in Equation~\ref{cohom} which runs in the opposite direction of the
filtration~\cite{DBLP:journals/corr/abs-1107-5665,DBLP:journals/dcg/SilvaMV11,DBLP:conf/compgeom/DeyFW14}. 
The birth of a cohomology class coincides
with the death of a homology class and the death of a cohomology
class coincides with the birth of a homology class.
Therefore, by tracking a cohomology basis along the filtration
direction and switching the notions of births and deaths, one
can obtain all information about the persistent homology of the complex.
The algorithm of de Silva et al.~\cite{DBLP:journals/dcg/SilvaMV11} computes the 
persistent cohomology following this principle which is reported
to work better in practice than the original persistence
algorithm~\cite{DBLP:journals/corr/abs-1107-5665}. Recently, 
Dey et al.~\cite{DBLP:conf/compgeom/DeyFW14} recognized that
tracking cohomology bases provides a simple and natural
extension of the persistence algorithm for filtrations
connected with general simplicial maps (and not simply inclusion). Their algorithm is based on
the notion of annotation~\cite{DBLP:conf/swat/BusaryevCCDW12} and, 
when restricted to only inclusions,
is a re-formulation of the algorithm of de Silva et al.~\cite{DBLP:journals/dcg/SilvaMV11}. 
Here we follow this annotation based algorithm.
%
\section{Persistent Cohomology Algorithm and Annotations}
\label{sec:annotations}
In this section, we recall the annotation-based persistent cohomology algorithm 
of~\cite{DBLP:conf/compgeom/DeyFW14}. It maintains a cohomology basis under
simplex insertions, where representative cocycles are maintained by the
value they take on the simplices.
We rephrase the description of this algorithm with coefficients in an
arbitrary field $\Field$, and use standard field notations $\langle \Field, +,
\cdot, -,/, 0, 1\rangle$.
\begin{definition}
Given a simplicial complex $\K$, let $\K^p$ denote the set of
$p$-simplices in $\K$. An annotation for $\K^p$ is an assignement
$\ann^p : \K^p \rightarrow \Field^g$ of an $\Field$-vector $\ann_\sigma =
\ann ^p(\sigma)$ of same length $g$ for each $p$-simplex $\sigma \in
\K^p$. We use $\ann$ when there is no ambiguity on the dimension. We
also have an induced annotation for any $p$-chain $c = \sum_i f_i \sigma_i$ given
by linear extension: $\ann_{c} = \sum_i f_i \cdot \ann_{\sigma_i}$.
\end{definition}
\begin{definition}
An annotation $\ann: \K^p\rightarrow \Field^g$ is {\em valid} if:\\
1. $g=\rank\, H_p(\K)$ and 2. two $p$-cycles $z_1$ and $z_2$  have
$\ann_{z_1}=\ann_{z_2}$ iff their homology
classes $[z_1]$ and $[z_2]$ are identical.
\end{definition}
\begin{proposition}[\cite{DBLP:conf/compgeom/DeyFW14}]
The following two statements are equivalent:\\
1. An annotation $\ann : \K^p \rightarrow \Field^g$ is valid\\
2. The cochains $\{ \phi_j \}_{j=1\cdots g}$ given by $\phi_j(\sigma)
  = \ann_\sigma [j]$ for all $\sigma \in \K^p$ are cocycles whose
  cohomology classes $\{ [\phi_j] \}_{j=1 \cdots g}$ constitute a basis
  of $H^p(\K)$.
\label{prop:cocycles}
\end{proposition}

A valid annotation is thus a way to represent a cohomology basis.
The
algorithm for computing persistent cohomology consists in maintaining a valid
annotation for each dimension when inserting all simplices in the order of the
filtration.
Since we process the filtration in a direction
opposite to the cohomology sequence (as in Equation~\ref{cohom}),
we discover the death points of cohomology classes earlier
than their birth points. To avoid confusion, we still say
that a new cocycle (or its class) is born when we discover it
for the first time and an existing cocycle (or its class) dies when
we see it no more.

We present the algorithm and refer
to~\cite{DBLP:conf/compgeom/DeyFW14} for its validity. We insert
simplices in the order of the filtration. Consider an elementary inclusion
$\K_i \hookrightarrow \K_{i}\cup \{\sigma\}$, with $\sigma$ a $p$-simplex.
Assume that to every simplex $\tau$ of any dimension in $\K_i$ is attached an
annotation vector $\ann_\tau$ from
a valid annotation $\ann$ of $\K_i$. We describe how to obtain
a valid annotation for $\K_{i}\cup \{\sigma\}$ from that of $\K_i$. 
We compute the
annotation $\ann_{\partial \sigma}$ for the boundary $\partial\sigma$
in $\K_i$ and take actions as follows:

\noindent {\bf Case 1}: If $\ann_{\partial\sigma} = 0$, $ g \leftarrow
g+1$ and the annotation vector of any $p$-simplex $\tau \in \K_i$ is
augmented with a $0$ entry so that $\ann_\tau = [f_1, \cdots
,f_g]^T$ becomes $[f_1, \cdots ,f_g,0]^T$. We assign to the new
simplex $\sigma$ the annotation vector $\ann_\sigma =
[0,\cdots,0,1]^T$. According to Proposition~\ref{prop:cocycles}, this
is equivalent to creating a new cohomology class represented by
$\phi(\tau)=0$ for $\tau\neq \sigma$ and $\phi(\sigma)=1$.

\noindent {\bf Case 2}: If $\ann_{\partial\sigma} \neq 0$, we consider
the non-zero element $c_j$ of $\ann_{\partial\sigma}$ with maximal
index $j$.  We now look for annotations of those $(p-1)$-simplices
$\tau$ that have a non-zero element at index $j$ and process them as follows.  If
the element of index $j$ of $\ann_{\tau}$ is $f \neq 0$, we add $-f / c_j \cdot \ann_{\partial\sigma}$
to $\ann_{\tau}$. Note that, in the annotation matrix whose columns
are the annotation vectors, this implements simultaneously a series of elementary
row operations,
where each row $\phi_i$ receives $\phi_i \leftarrow \phi_i -
(\ann_{\partial\sigma}[i]/c_j)\times \phi_j$. As a result, all the
elements of index $j$ in all columns are now
$0$ and hence the entire row $j$ becomes $0$. We then remove the row $j$
and set $g \leftarrow g-1$. $\sigma$ is assigned $\ann_{\sigma}=0$. According to Proposition~\ref{prop:cocycles},
this is equivalent to removing the $j^{\text{th}}$ cocycle
$\phi_j(\tau) = \ann_\tau [j]$.

As with the original persistence algorithm,
the pairing of simplices is derived from the creation and
destruction of the cohomology basis elements.
%
\section{Data Structures and Implementation}
\label{sec:implementation}
In this section, we present our implementation of the annotation-based
persistent cohomology algorithm. We separate the representation of the
simplicial complex from the representation of the cohomology groups. 
\begin{figure}[t]
\centering
\includegraphics[width=13cm]{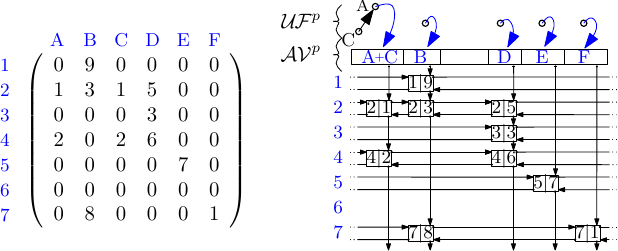}
\caption{Compressed annotation matrix of a matrix with integer coefficients.}
\label{fig:sparsemat_rep}
\end{figure}
%
\subsection{Representation of the Simplicial Complex}
\label{subsec:SC}
We represent the simplicial complex $\K$ in a data structure $\SC$ equipped
with the operation  {\sc Compute-boundary}$(\sigma)$ that  computes the boundary of a simplex
  $\sigma$. We denote by $\C_\partial^p$ the complexity of this
  operation where $p$ is the dimension of $\sigma$. Additionally, the simplices are ordered according to the filtration.

Two data structures to represent simplicial complexes are of
particular interest here. The first
one is the \emph{Hasse diagram}, which is the graph whose nodes are in
bijection with the 
simplices (of all dimensions) of the simplicial complex and an edge
links two nodes representing two simplices $\tau$ and $\sigma$ iff 
$\tau \subseteq \sigma$ and the dimensions of $\tau$ and $\sigma$ differ by $1$.
The second data structure is the \emph{simplex
  tree} introduced in~\cite{DBLP:journals/algorithmica/BoissonnatM14}, which is a specific
spanning tree of the Hasse diagram. 
For a simplicial complex $\K$ of dimension $k$ and a simplex $\sigma
\in \K$
of dimension $p$, the Hasse diagram has
size $O(k|\K|)$ and allows to compute {\sc Compute-boundary}$(\sigma)$
in time $\C_\partial^p = O(p)$, whereas the simplex tree has size $O(|\K|)$ and
allows to compute {\sc Compute-boundary}$(\sigma)$ in
time $\C_\partial^p = O(p^2D_{\text{m}})$, where $D_{\text{m}}$ is typically a small value related to the time needed
 to traverse the simplex tree. 
Both structures can be used in our setting. For readability, we will use a Hasse diagram in the following.
\subsection{The Compressed Annotation Matrix}
For each dimension $p$, the $p^{\text{th}}$ cohomology
group can be seen as a valid annotation for the $p$-simplices of
the simplicial complex. 
 Hence, an annotation $\ann : \K^p \rightarrow
\Field^g$ can be represented as a $g \times |\K^p|$ matrix with 
elements in $\Field$, where each column is an annotation vector associated to a $p$-simplex. We describe how to represent this
annotation matrix in an efficient way.

\medskip

\noindent
{\bf Compressing the annotation matrix:}
In most applications, the annotation matrix is sparse and we store it 
as illustrated
in Figure~\ref{fig:sparsemat_rep}. A column is represented as the singly-linked
list of its non-zero elements, where the list contains a pair $(i,f)$
if the $i^{th}$ element of the column is $f \neq 0$. The pairs in the
list are ordered according to row index $i$.
All pairs $(i,f)$  with same row index $i$ are linked in a
doubly-linked list.

\medskip

\noindent
{\bf Removing duplicate columns:} (see Figure~\ref{fig:sparsemat_rep})
To avoid storing duplicate columns, we use two data structures. The
first one, $\AV^p$, stores the annotation vectors and allows fast
search, insertion and deletion. $\AV^p$ can be implemented as a
red-black tree or a hash table. We denote by $\C_{\AV}^p$
the complexity of an operation in $\AV^p$. For example, if $\AV^p$
contains $n$ elements and $c_{\text{max}}$ is the length of the
longest column, we have $\C_{\AV}^p = O(c_{\text{max}} \log (n))$ for a
red-black tree implementation and $\C_{\AV}^p = O(c_{\text{max}})$ amortized for a hash-table. The simplices of the same
dimension that have the same annotation vector are now stored in a
same set and the various (and disjoint) sets are stored in a \emph{union-find}
data structure denoted $\UF^p$.  $\UF^p$ is encoded as a forest where each tree
contains the elements of a set, the root being the
``representative'' of the set. The trees of $\UF^p$ are in bijection
with the different annotation vectors stored in 
$\AV^p$ and  the root of each tree maintains a pointer to
the corresponding annotation vector in $\AV^p$. Each node
representing a $p$-simplex $\sigma$ in the simplicial
complex $\SC$ stores a pointer to an element of the tree of $\UF^p$ associated to the annotation vector $\ann_{\sigma}$. Finding the
annotation vector of $\sigma$ consists in getting the element it
points to in a tree of $\UF^p$ and then finding the root of the
tree which points to $\ann_{\sigma}$ in $\AV^p$. We avail
the following operations on $\UF^p$:

$\bullet$ {\sc Create-set}: creates a new tree containing
 one element.

$\bullet$ {\sc Find-root}: finds the root of a tree, given
  an element in the tree.

$\bullet$ {\sc Union-sets}: merges two trees.

The number
of elements maintained in $\UF^p$ is at most the number of simplices
of dimension $p$, i.e. $|\K^p|$. The operations {\sc Find-root} and
{\sc Union-sets} on $\UF^p$ can be computed in amortized time $O(\alpha(|\K^p|))$, where $\alpha(\cdot)$ is the very slowly
growing inverse Ackermann function (constant less than $4$ in
practice), and {\sc Create-set} is performed in constant time.
We will refer to this data structure as the \emph{Compressed Annotation Matrix}.

\medskip

\noindent
{\bf Operations:}
The compressed annotation matrix described above supports the following operations.
We define $c_{\text{max}}$ to be the maximal number of non-zero elements
in a column of the compressed annotation matrix (or equivalently in an
annotation vector) and $r_{\text{max}}$ to be the maximal number of
non-zero elements in a row of the compressed annotation matrix, during
the computation. We will express our complexities using $c_{\text{max}}$
and $r_{\text{max}}$:
%

$\bullet$ {\sc Sum-ann}$(\ann_1,\ann_2)$: computes the sum of two
  annotation vectors $\ann_1$ and $\ann_2$, and returns the lowest
  non-zero coefficient if it exists. The column elements are sorted by increasing row
  index, so the sum is performed in $O(c_{\text{max}})$ time.

$\bullet$ {\sc Search-ann}/{\sc Add-ann}/{\sc Remove-ann} $(\ann)$:  searches, adds
  or removes an annotation vector $\ann$ from $\AV^p$ in $O(\C_{\AV}^p)$ time.

$\bullet$ {\sc Create-cocycle}$()$: implements {\bf Case 1} of the
algorithm described in section~\ref{sec:annotations}. It inserts a new
column in $\AV^p$ containing one element $(i_{\text{new}},1)$, where
$i_{\text{new}}$ is the index of the created cocycle. This is
performed in time $O(\C_{\AV}^p)$.  We also create a new disjoint set
in $\UF^p$ for the new column. This is done in $O(1)$ time using {\sc
  Create-set}. {\sc Create-cocycle}$()$ takes $O(\C_{\AV}^p)$ in total.

$\bullet$ {\sc Kill-cocycle}$(\ann_{\partial \sigma},c_j,j)$: implements
  {\bf Case 2} of the algorithm. It finds all columns with a non-zero
  element at index $j$ and, for each such column $A$, it adds to $A$ the
  column $-f / c_j \cdot \ann_{\partial\sigma}$ if $f$ is the non-zero
  element at index $j$ in $A$. To find the columns with a non-zero
  element at index $j$, we use the doubly-linked list of row $j$. We
call {\sc Sum-ann} to compute the sums. The overall time needed for
all columns is $O(c_{\text{max}}\, r_{\text{max}})$ in the worst-case. Finally, we remove duplicate columns using operations
on $\AV^p$ (in $O(r_{\text{max}}\, \C_{\AV}^{p-1})$ time in the worst-case) and call {\sc Union-sets} on $\UF^{p-1}$ if two sets of
simplices, which had different annotation vectors before calling {\sc
  Kill-cocycle}, are assigned the same annotation vector. This is
performed in at most $O(r_{\text{max}}\, \alpha(|\K^{p-1}|) )$ time. The
total cost of {\sc Kill-cocycle} is $O(r_{\text{max}} ( c_{\text{max}}
+\C_{\AV}^{p-1} + \alpha(|\K^{p-1}|)))$. 
\subsection{Computing Persistent Cohomology}
Given as input a filtered simplicial complex represented in a data
structure $\SC$, we compute its persistence diagram.

\medskip

\noindent
{\bf Implementation of the persistent cohomology algorithm:}
We insert the simplices in the filtration order and update the data
structures during the successive insertions. The simplicial complex $\K$ is
stored in a simplicial complex data structure $\SC$ and we maintain, for each dimension
$p$, a compressed annotation matrix, which is empty at
the beginning of the computation. For readability, we add the following operation on the set of data structures:

$\bullet$ {\sc Compute-${\ann_{\partial \sigma}}$}($\sigma$): given a
$p$-simplex $\sigma$ in $\K$, computes its boundary in $\SC$ using {\sc
  Compute-boundary} (in
$O(\C_\partial^p)$ time). For each of the $p+1$ simplices in $\partial \sigma$, it
then finds
their annotation vector using  {\sc
  Find-root} in $\UF^{p-1}$ (in $O(p \alpha(|\K^{p-1}|))$
time). Finally, it sums all these annotation vectors together (with
the appropriate $+/-$ sign) using at most $p+1$ calls to  {\sc
  Sum-ann} (in $O(p\, g_{\text{m}})$ time). Note that, with the compression method,
two simplices in $\partial \sigma$ may point to the same annotation
vector; the computation is accelerated by adding such annotation vector
only once, with the appropriate multiplicative coefficient. The total
worst case complexity of this operation is $O(\C_\partial^p+p\, 
\alpha(|\K^{p-1}|)+p\, g_{\text{m}})$.

Let $\sigma$ be a $p$-simplex to be inserted. We compute the annotation
vector of $\partial \sigma$ using  {\sc Compute-${\ann_{\partial
      \sigma}}$}. Depending on the value of $\ann_{\partial \sigma}$, we call either {\sc
  Create-cocycle} or {\sc Kill-cocycle}. The algorithm computes the pairing of simplices from which one can deduce the persistence diagram. By reversing the pointers from the $\UF^p$s to the simplices in $\SC$, one can compute explicitly the representative cocycles of the basis classes and have an explicit representation of the cohomology groups along the computation.

\medskip

\noindent
{\bf Complexity analysis:} Let $k$ be the dimension and $m$ the
number of simplices of $\K$. Recall that $c_{\text{max}}$
and $r_{\text{max}}$ represent respectively the maximal number of non-zero elements
in an annotation vector and in a row of the compressed annotation
matrix, along the computation. Recall that, in dimension $p$,
$\C_\partial^p$ is the complexity of {\sc Compute-boundary} in
$\SC$ and $\C_{\AV}^p$
the complexity of an operation in $\AV^p$. $\alpha(\cdot)$ is the inverse Ackermann function.

The complexity for inserting $\sigma$ of
dimension $p$ is: 
$$\displaystyle O\left( \C_\partial^p + p(\alpha(|\K^{p-1}|) + c_{\text{max}}) +
  \C_{\AV}^p  + r_{\text{max}}(c_{\text{max}}
  + \C_{\AV}^{p-1} + \alpha(|\K^{p-1}|) )   \right)$$
Consequently, the total cost for computing the persistent
cohomology 
is:
$$\displaystyle O\left(m \times \left[  \C_\partial^k + k(\alpha(m) + c_{\text{max}}) + r_{\text{max}}(c_{\text{max}}
  + \C_{\AV} + \alpha(m) )   \right] \right)$$
Specifically, if we implement $\SC$ as a Hasse diagram and the
$\AV$s as hash-tables, we get $\C_\partial^k=O(k)$ and $\C_{\AV} =
O(c_{\text{max}})$. If we consider $\alpha(m)$ as a small constant
and remove it for
readability, we get that the total cost for computing persistent
cohomology is:
$$\displaystyle O(m\, c_{\text{max}} (k+r_{\text{max}}))$$
We show in
section~\ref{sec:experiments} that $c_{\text{max}}$ and $r_{\text{max}}$
remain small in practice. Hence, the practical complexity of the
algorithm is linear in $m$ for a fixed dimension.
%
%
\section{Reordering Iso-simplices}
\label{sec:lazy}
Many simplices, called iso-simplices, may have the same filtration
value. This situation is common when the filtration is induced by a
geometric scaling parameter. Assume that we want to compute the
cohomology groups $H^p(\K_{i+1})$ from $H^p(\K_i)$ where $\K_i
\subseteq \K_{i+1}$ and all simplices in $\K_{i+1} \setminus \K_i$
have the same filtration value. Depending on the insertion order of
the simplices of $\K_{i+1} \setminus \K_i$, the dimension of the
cohomology groups to be maintained along the computation may vary a
lot as well as the computing time. This may lead to a
computational bottleneck. We propose a heuristic to reorder
iso-simplices and show its practical efficiency in Section~\ref{sec:experiments}.

Intuitively, we want to avoid the creation of many ``holes'' of
dimension $p$ and want to fill them up as soon as possible with simplices of
dimension $p+1$. For example, in Figure~\ref{fig:up_down}, we want to
avoid inserting all edges first, which will create two holes that
will be filled when inserting the triangles. To do so,
we look for the maximal faces to be inserted and recursively insert
their subfaces. We conduct the recursion so as to minimize the maximum
number of holes. In addition, to avoid the creation of holes due to
maximal simplices that are incident, maximal simplices sharing
subfaces are inserted next to each other.
\begin{figure}[!t]
\centering
\includegraphics[width=16cm]{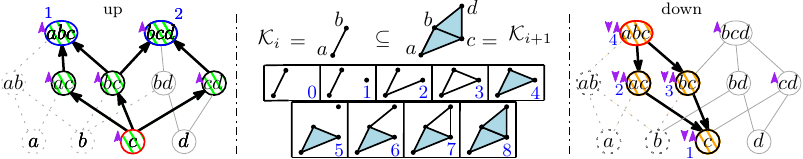}
\caption{Inclusion $\K_i \subseteq \K_{i+1}$. Left: upward traversal (in
  green) from simplex $\{c\}$. The ordering of the maximal cofaces
  appears in blue. Right: downward traversal (in orange) from
  simplex $\{abc\}$. The ordering of the subfaces appears in blue.}
\label{fig:up_down}
\end{figure}
We can describe the reordering algorithm in terms of a
graph traversal. The graph considered is the graph of the Hasse
diagram of $\K_{i+1} \setminus
\K_i$, defined in section~\ref{subsec:SC} (see
Figure~\ref{fig:up_down}).

Let $\sigma_1 \cdots \sigma_\ell$ be the iso-simplices of  $\K_{i+1}\setminus \K_i$,
 sorted so as to respect the inclusion order. We attach to each simplex two flags, a flag $F_{\text{up}}$ and a flag
$F_{\text{down}}$, set to {\sc 0} originally. When inserting a simplex
$\sigma_j$, we proceed as follows. We traverse the Hasse diagram
upward in a depth-first fashion and list the inclusion-maximal cofaces
of $\sigma_j$ in $\K_{i+1} \setminus \K_i$.  The flags $F_{\text{up}}$
of all traversed nodes are set to {\sc 1}
and the maximal cofaces are ordered according to the traversal. From each
maximal coface in this order, we then traverse the graph downward and order the
subfaces in a depth-first fashion: this last order will be the order
of insertion of the simplices. The flags $F_{\text{down}}$ of all traversed nodes
are set to {\sc 1}.  We stop the upward (resp. downward) traversal when
we encounter a node whose flag $F_{\text{up}}$ (resp. $F_{\text{down}}$) is set to {\sc 1}. We do
not insert either simplices that have been inserted previously.

By proceeding as above on all simplices of the sequence $\sigma_1
\cdots \sigma_\ell$, we define a new ordering which respects the
inclusion order between the simplices. Indeed, as the downward
traversal starts from a maximal face and is depth first, a face is
always inserted after its subfaces.  Every edge in the graph is
traversed twice, once when going upward and the other when going
downward.  Indeed, during the upward traversal, at each node $N$
associated to a simplex $\sigma_N$, we visit only the edges between
$N$ and the nodes associated to the cofaces of $\sigma_N$ and, during
the downward traversal, we visit only the edges between $N$ and
the nodes associated to the subfaces of $\sigma_N$. If
$\K_{i+1}\setminus \K_i$ contains $\ell$ simplices, the reordering
takes in total $O(\ell \times (\C_{\partial} + \C_{co\partial}))$
time, where $\C_{\partial}$ (resp. $\C_{co\partial}$) refers to the
complexity of computing the codimension $1$ subfaces (resp. cofaces)
of a simplex in the simplicial complex data structure $\SC$.
The reordering of the filtration can either be done as a preprocessing
step if the whole filtration is known, or on-the-fly  as only the
neighboring simplices  of a simplex need
to be known at a time. The reordering of a set of iso-simplices respects the inclusion order of
the simplices and the filtration, and therefore does not change the persistence diagram of the
filtered simplicial complex. This is a direct consequence of the stability
theorem of persistence
diagrams~\cite{DBLP:journals/dcg/Cohen-SteinerEH07}. 
However, it may change the pairing of simplices.
%
\section{Experiments}
\label{sec:experiments}
In this section, we report on the experimental performance of our
implementation. The compressed annotation matrix implementation of persistent cohomology 
(denoted by {\tt CAM} in the following) is part of the {\tt Gudhi} library~\cite{gudhi:FilteredComplexes,gudhi:PersistentCohomology,gudhilibrary_ICMS14,gudhi:urm}.

Given a filtered simplicial complex as input,
we measure the time taken by our implementation to compute its
persistent cohomology, and provide various statistics.
We compare the timings with state-of-the-art software computing persistent homology and cohomology. Specifically, we compare our
implementation with the {\tt Dionysus} library~\cite{dionysus_morozov} which provides
implementation for persistent
homology~\cite{DBLP:journals/dcg/EdelsbrunnerLZ02,DBLP:journals/dcg/ZomorodianC05}
and persistent cohomology~\cite{DBLP:journals/dcg/SilvaMV11} (denoted {\tt
  DioCoH}) with field coefficients in $\Z_p$, for any prime $p$. We
also compare our implementation with the {\tt PHAT} library (version 1.4)~\cite{phat_lib,DBLP:conf/icms/BauerKRW14} which provides an implementation of
the optimized algorithm for persistent
homology~\cite{Bauer:arXiv1303.0477,Chen11persistenthomology} (using
the {\tt --twist} option) as well as an implementation of persistent
cohomology~\cite{Bauer:arXiv1303.0477,DBLP:journals/corr/abs-1107-5665} (using the {\tt
  --dualize} option), with coefficients in $\Z_2$
only. {\tt DioCoH} and {\tt PHAT} have been reported to be the most
efficient implementation in
practice~\cite{Bauer:arXiv1303.0477,DBLP:journals/corr/abs-1107-5665}. All timings are measured on a Linux machine with $3.00$ GHz processor and $32$
GB RAM. {\tt Dionysus}, {\tt PHAT} and our implementation are
written in {\tt C++} and compiled with {\tt gcc 4.6.2} with
optimization level {\tt -O3}. Timings are all averaged over $10$ independent
runs. The symbols $\text{T}_\infty$ means that the computation lasted more than
$12$ hours.

\begin{figure}[t]
\centering
  \setlength{\tabcolsep}{8.2pt}
  \begin{tabular}{cccc}
    \hspace{6.0cm} {\tt DioCoH} & {\tt PHAT}$^\perp$ & {\tt PHAT} & {\tt CAM} \\
  \end{tabular}
  \setlength{\tabcolsep}{2.0pt}
  \begin{tabular}{| l l r r r r r r | r r | r r | r r | r r |}
    \hline
    Data & Cpx &  $|\mathcal{P}|$& $D$& $d$&$\rho_{\text{max}}$&
    $k$&$|\K|$ & $\Z_2$ & $\Z_{11}$ & $\Z_2$ & $\Z_{11}$ & $\Z_2$ & $\Z_{11}$ & $\Z_2$ & $\Z_{11}$ \\
    \hline

    {\bf Cy8}& {\tt Rips} & $6040$&$24$&$2$&$0.41$&$16$&$21\times 10^6$
    & $420$ & $4822$ & $44$  & $-$ & $5.4$ & $-$ & $6.4$ & $6.5$ \\

    {\bf S4}& {\tt Rips} & $507$ & $5$ & $4$  & $0.715$  & $5$ &
    $72\times 10^6$  & $943$ & $1026$ & $96$ & $-$  & $3554$ & $-$ & $22.5$ & $23.2$ \\

{\bf L57} & {\tt Rips} & $4769$ & $-$ & $3$ & $0.02$ & $3$ & $34\times

10^6$ & $239$ & $240$ & $35$ & $-$ & $970$ & $-$ & $9.3$ & $9.5$ \\

    {\bf Bro}& {\tt Wit}  & $500$   & $25$  &  ?  & $0.06$  & $18$  &
    $3.2\times 10^6$  & $807$ & $\text{T}_{\infty}$ & $6.5$ & $-$ & $0.88$ & $-$ & $2.7$ & $2.9$ \\

    {\bf Kl}& {\tt Wit} & $10000$ & $5$ & $2$ & $0.105$ & $5$ & $74

    \times 10^6$ & $569$ & $662$ & $100$ & $-$ & $1771$ & $-$ & $19.7$ & $19.9$ \\

{\bf L35} & {\tt Wit} & $700$ & $-$ & $3$ & $0.06$ & $3$ & $18\times

10^6$ &$109$&$110$&$18$&$-$&$865$&$-$& $5.1$ & $5.1$ \\

    {\bf Bud}& {\tt $\alpha$Sh}  & $49990$  & $3$  & $2$  & $\infty$  &
    $3$ & $1.4\times 10^6$ & $30.0$ & $30.9$ & $2.5$ & $-$ & $0.33$ & $-$ & $0.7$ & $0.7$ \\
    
    {\bf Nep} & {\tt $\alpha$Sh} & $2\times
    10^{6}$&$3$&$2$&$\infty$&$3$&$57\times 10^6$ & $\text{T}_{\infty}$ &
    $\text{T}_{\infty}$ & $166$ & $-$ & $33$ & $-$ &
    $39.5$ & $40.2$ \\
    \hline
  \end{tabular}
  \caption{Data, timings (in seconds) and statistics.}
  \label{fig:table_time}
\end{figure}
We construct three families of simplicial complexes~\cite{DBLP:books/daglib/0025666} which are
of particular interest in topological data analysis: the Rips
complexes (denoted {\tt Rips}), the relaxed witness complexes (denoted
{\tt Wit}) and the
$\alpha$-shapes (denoted {\tt $\alpha$Sh}). These complexes depend on
a relaxation parameter $\rho$. When the data points are embedded, the
complexes are constructed up to embedding dimension, with euclidean
metric. They are constructed up to the intrinsic dimension of the
space with intrinsic metric otherwise. We use a variety of both real and synthetic datasets:
{\bf Cy8} is a set of
points in $\R^{24}$, sampled from the space of conformations of the
cyclo-octane molecule, which is the union of two
intersecting surfaces;
 {\bf S4} is a set of points sampled from the unit $4$-sphere in
$\R^5$; {\bf L57} and {\bf L35} are sets of points in the \emph{lens spaces}
$L(5,7)$ and $L(3,5)$ respectively, which are non-embedded spaces;
{\bf Bro} is a set of $5\times 5$ {\it high-contrast patches}
derived from natural images, interpreted as vectors in $\R^{25}$, from
the Brown database;
 {\bf Kl} is a set of points sampled from the
surface of the figure eight Klein Bottle embedded in $\R^5$;
 {\bf Bud} is a
set of points sampled from the surface of the {\it Happy Buddha} (\url{http://graphics.stanford.edu/data/3Dscanrep/})
in $\R^3$;
and {\bf Nep} is a set of points sampled from the surface of the {\it
  Neptune statue} (\url{http://shapes.aimatshape.net/}). Datasets are listed in Figure~\ref{fig:table_time} with
details on the sets of points $\mathcal{P}$,
their size $|\mathcal{P}|$, the ambient dimension $D$, the intrinsic
dimension $d$ of the object the sample points belong to (if known),
the threshold $\rho_{\text{max}}$, the dimension $k$ of the simplicial
complexes and the size $|\K|$ of the simplicial complexes.

\medskip

\noindent
{\bf Time Performance:} As {\tt Dionysus} and {\tt PHAT} encode
explicitely the boundaries of the simplices, we use a Hasse diagram
for implementing $\SC$. We thus have the same time complexity for
accessing the boundaries of simplices. 
We use the persistent homology
algorithm of {\tt PHAT} with options {\tt --twist --sparse\_pivot\_column} and the
persistent cohomology algorithm (noted {\tt PHAT}$^{\perp}$) with
option {\tt --twist --sparse\_pivot\_column --dualize}. The {\tt --twist} algorithm is the most 
efficient for our experiments. The {\tt --sparse\_pivot\_column} representation for the "pivot column" is the most efficient 
representation, for our experiemnts, that can generalize to arbitrary coefficients for persistent 
homology\footnote{{\tt PHAT} also contains 
pivot column representations specifically optimized for $\Z_2$ coefficients using low level bit operations. These bit representations 
usually accelerate by a constant factor the computation of persistence on complexes involving a lot of arithmetic operations, in particular the persistent homology algorithm 
{\tt PHAT} on {\bf S4} and {\bf Kl} in Figure~\ref{fig:table_time}. They do not change the asymptotic analysis done in this paragraph.}.
%
%
%
As illustrated in Figure~\ref{fig:table_time}, the
persistent cohomology algorithm of {\tt
  Dionysus} is always several times slower than our
implementation. 
Moreover, {\tt DioCoH} is sensitive to the field used, as illustrated
in the case of {\bf Cy8} and {\bf Bro}. On the contrary, {\tt CAM} shows almost identical performance for $\Z_2$
and $\Z_{11}$ coefficients on all examples. 
The persistent cohomology algorithm {\tt
PHAT}$^{\perp}$ performs better than {\tt DioCoH}.
However, {\tt CAM} is still between $2.3$ and $6.9$ times faster.

The persistent homology algorithm of {\tt PHAT} shows good performance
in the case of the alpha shapes and on {\bf Cy8} and {\bf Bro}: {\tt
  CAM} and {\tt PHAT} have close timings.  
However, {\tt CAM} scales better to more complex examples
(such as {\bf S4}, {\bf L57}, {\bf Kl} and {\bf L35}, which have higher
intrinsic dimension and more complex topology). Indeed, the running
time per simplex of {\tt CAM} remains stable on all examples and for
all field coefficients (between $2.7\times 10^{-7}$ and $9.1\times 10^{-7}$ seconds per simplex).

\begin{figure}[t]
  \centering
  \setlength{\tabcolsep}{3.7pt}
  \begin{tabular}{|l || l |}
\hline
    \setlength{\tabcolsep}{3.5pt}
   \begin{tabular}{r|r|r}
      {\bf Nep} &  $|M|$ & $\#\Field op.$  \\
      \hline
      Compression & $126057$ & $84\times 10^6$ \\
      $\neg$Compression & $574426$ & $3860 \times 10^6$ \\
    \end{tabular}
    &
  \begin{tabular}{r|r|r}
      {\bf Nep} &  average & maximum  \\
      \hline
      $c_{\text{av}} , c_{\text{max}}$ & $0.79$ & $18$ \\
      $r_{\text{av}} , r_{\text{max}}$ & $1.02$ & $18$ \\
    \end{tabular}
\\
\hline
\end{tabular}

  \setlength{\tabcolsep}{4.33pt}
  \begin{tabular}{|l || l |}
\hline
  \begin{tabular}{r|r}
      {\bf Bro} &  time \\
      \hline
       Reordering & $2.9$ s.\\
      $\neg$Reordering & $14.2$ s. \\
    \end{tabular}

    &
    \begin{tabular}{r|ccc|ccc}
     {\bf Bro} & &$\Z_{11}$ & & &$\Q$& \\
      \hline
   &   $M_{\text{DS}}$ & $\ann_{\partial \sigma}$ &
      $M_{\text{op}}$ & $M_{\text{DS}}$ & $\ann_{\partial \sigma}$ &
     $M_{\text{op}}$ \\

    &  $71\%$ & $19\%$ & $10\%$ & $67\%$ & $21\%$ & $12\%$ \\
    \end{tabular}
    \\
\hline
  \end{tabular}
  \caption{Statistics on the effect of the optimizations.}
  \label{fig:table_stat}
\end{figure} 

\medskip

\noindent
{\bf Statistics and Optimization:}
Figure~\ref{fig:table_stat} presents statistics about the
computation. %
The top table presents, on the left, the effect of the
compression (removal of duplicate columns) of
the annotation matrix on the number of elements $|M|$ stored in the
sparse representation and the number of changes $\#\Field op.$ in the
matrix during the computation of the persistence diagram of {\bf Nep}. We note a reduction factor of
$4.5$ for the size of the matrix, and we proceed to $46$ times less
field operations with the compression. Considering {\bf Nep} is $57$
million simplices, we proceed to less than $1.5$ field operations per simplex on average. The right part of the table shows the average and
maximum number of non-zero elements in a column when proceeding to a
sum of annotation vectors ({\sc Sum-ann}) and the average and maximum
number of non-zero elements in a row when proceeding to its reduction
({\sc Kill-cocycle}). These values are key
  variables ($c_{\text{max}}$ and $r_{\text{max}}$ respectively) in the complexity analysis of
  the algorithm. We note that these values remain really small.
%
%
The bottom table presents the effect of the reordering strategy on the
example {\bf Bro}. We note that reordering iso-simplices makes the computation $4.9$ faster. Finally, the right side of the table
presents how the computing time is divided into
maintaining the compressed annotation matrix (noted
$M_{DS}$), computing the annotation vector $\ann_{\partial \sigma}$
and modifying the values of the elements in the compressed annotation
matrix (noted $M_{\text{op}}$). The percentage are given when computing
persistent cohomology with $\Z_{11}$ and $\Q$
coefficients. The computational complexity of field operations $\langle \Field, +, \cdot,
-,/, 0, 1\rangle$ depends on the field we use. For $\Z_{11}$, or any
field of small cardinality,  the operations can be precomputed and
accessed in constant time. The
field operations in $\Q$ are more costly. Specifically, an element $q$ in $\Q$ is represented as a pair of coprime integers $(r,s)$ such that $q=r/s$, and field operations may require gcd computation to ensure that nominator and denominator remain coprime.
However, the computational
time of {\tt CAM} is quite
insensitive to the field we use.
Specifically, as it minimizes the number of matrix changes using the
compression method, the
computational time is only increased by $8\%$ when computing
persistence with $\Q$ coefficients instead of $\Z_{11}$, whereas the
computation involving field operations takes $34\%$ more time.

In all our experiments, the size of the compressed annotation matrix
is negligible compared to the size of the simplicial
complex. Consequently, combined with the simplex tree data structure~\cite{DBLP:journals/algorithmica/BoissonnatM14} for representing the simplicial complex, we have been able to compute
the persistent cohomology of simplicial complexes of several hundred
million simplices in high dimension.


\medskip

\noindent
{\bf Acknowledgement:}
The research leading to these results has received funding from the
European Research Council (ERC)
under  the  European  Union's  Seventh  Framework  Programme
(FP/2007-2013)  /  ERC  Grant
Agreement No. 339025 GUDHI (Algorithmic Foundations of Geometry
Understanding in Higher Dimensions).

This research is also partially supported by the 7th Framework Programme for Research of the European Commission, under FET-Open grant number 255827 (CGL Computational Geometry
Learning). This research is also partially supported by NSF (National Science Foundation, USA)
grants CCF-1048983 and CCF-1116258.
%
%
\newpage
\bibliographystyle{plain}
\bibliography{bibliography}

%

\end{document}